\pdfoutput=1
\documentclass[pdflatex,sn-mathphys-num]{sn-jnl}

\usepackage{graphicx}
\usepackage{amsmath,amssymb,amsfonts}
\usepackage{amsthm}
\usepackage[title]{appendix}
\usepackage{xcolor}
\usepackage{booktabs}
\usepackage{array}
\usepackage{multirow}

\theoremstyle{thmstyleone}
\newtheorem{lemma}{Lemma}
\theoremstyle{thmstyletwo}

\newcommand{\nEB}{n_{\mathrm{EB}}}
\newcommand{\Tr}{\mathrm{Tr}}
\newcommand{\lmin}{\lambda_{\min}}
\newcommand{\cnot}{\textsc{cnot}}
\newcommand{\ket}[1]{|#1\rangle}
\newcommand{\bra}[1]{\langle #1|}

\begin{document}

\title[Measuring an entanglement-breaking index]{Measuring the entanglement-breaking index of a qubit collision channel on a quantum processor: protocol and worked example}

\author*[1]{\fnm{Eran} \sur{Kopel}}\email{erankopel@tauex.tau.ac.il}
\affil*[1]{\orgname{Tel Aviv University}, \orgaddress{\city{Tel Aviv}, \country{Israel}}}

\abstract{The entanglement-breaking index of a quantum channel is the smallest number of
self-compositions after which the channel destroys all entanglement with any reference
system. A photonic experiment has fixed an index of 2; to our knowledge no larger index
has been located by measuring the $n$-round channel across rounds. We give a protocol for
doing so on a programmable processor, for any qubit channel realised as a collision model
whose ancillas are prepared in states diagonal in a known basis, so that the bath
polarisation is set exactly by weighting basis-state configurations in the analysis.
Single-qubit process tomography of the $n$-round channel gives the smallest eigenvalue of
its partially transposed Choi matrix, read against zero at a pre-registered five standard
deviations; entanglement breaking survives further rounds, so the readings are monotone in
$n$ and the index follows. Comparing each direct measurement with the composition of the
measured single-round map tests the repeated-interaction assumption. We show that a
readable step needs margins on both sides, show how depolarising and reset errors move the
readings, and give a power-based pilot rule. As a worked example we take one round of a
three-qubit feedback loop, compiled to eight CZ gates on a heavy-hex path with mid-circuit
reset, with certified ideal indices 3, 2 and 2 for three baths. The run is estimated to
cost a few hundred seconds of processor time, and simulations under a device noise model,
with the pilot's shot raise at low noise, recover all three indices.}


\maketitle

\section{Introduction}\label{sec:intro}

A quantum channel is entanglement breaking (EB) if its output is separable from any
reference system, whatever the input \cite{HSR2003,Ruskai2003}. A channel $\Phi$ that is
not EB may become so under composition, and the smallest $n$ for which it does,
\begin{equation}
  \nEB(\Phi) \;=\; \min\{\, n \ge 1 : \Phi^n \text{ is EB}\,\},
  \label{eq:index}
\end{equation}
is its entanglement-breaking index, or EB order
\cite{DePasquale2012,LG2015}. Which channels are eventually EB, and how fast, is well
studied \cite{Rahaman2018,Christandl2019,Hanson2020}. For a qubit channel the question
is a finite computation: $\Phi^n$ is EB exactly when the partial transpose of its Choi
matrix \cite{Choi1975} is positive semidefinite \cite{Peres1996,Horodecki1996,HSR2003}.

The experimental record is thinner. Cuevas et al.\ \cite{Cuevas2017} built a photonic
channel, certified by Choi-state tomography that its two-fold concatenation $\Phi\circ\Phi$
is EB, and showed that an intermediate unitary $F$ restores entanglement transmission,
$\Phi\circ F\circ\Phi$ not being EB. By the ideal property recalled in
Sec.~\ref{sec:choi}, the second result also shows that $\Phi$ itself is not EB, so their
data fix $\nEB(\Phi) = 2$, and they relate their construction to the EB order.
Certification that a single channel is not EB, the basic test for a quantum memory, has
been developed in theory \cite{Rosset2018} and demonstrated with photons
\cite{Mao2020,Graffitti2020}. Entanglement sudden death \cite{Almeida2007} is a different
phenomenon: a state loses its entanglement at a finite time under continuous noise, while
the local channel is not itself EB. To our knowledge no experiment has located a larger
index, that is, scanned the number of rounds beyond two and found the step at which the
$n$-round channel becomes EB, nor compared such a direct measurement with the composition
of the measured single-round map.

Collision models are the natural setting \cite{Scarani2002,Ciccarello2022,Campbell2021}.
A system meets a fresh ancilla, or a freshly reset one, in each round, so one round is one
compiled circuit block and $n$ rounds are its $n$-fold repetition. Collision models have
run on photonic, superconducting and trapped-ion hardware
\cite{Cuevas2019,GarciaPerez2020,Cattaneo2023,Wang2026}, and multi-step process
characterisation that tests the Markov assumption has been demonstrated on a superconducting
processor \cite{White2020}.

This paper gives a protocol that measures $\nEB$ directly, for any qubit channel realised
as a collision round whose ancillas are prepared in a state diagonal in a known basis
(Secs.~\ref{sec:setting} and \ref{sec:protocol}). Four features make it practical. The
bath polarisation is set exactly, without preparing mixed states, by weighting basis-state
configurations in the analysis (Lemma~\ref{lem:mix}). Single-qubit process tomography of
the $n$-round channel needs no reference qubit. Each reading is pre-registered at five
standard deviations, and the readings must be monotone in $n$, so an inconsistent pattern
is itself a signal. And the comparison of the direct $n$-round channel with the composition
of the measured single-round map tests the repeated-interaction assumption on which
collision models rest. Section~\ref{sec:design} gives the design rules: a step of the
staircase is readable only if both of its sides are, and noise thickens one side while
eating the other. Section~\ref{sec:example} works the protocol through for one round of a
three-qubit feedback loop \cite{PaperI,PaperII} on a heavy-hex superconducting processor, from
certified targets to compilation, budget and simulated runs.
Section~\ref{sec:beyond} summarises what the same machinery certifies beyond static powers
and why those targets are out of reach today, and Sec.~\ref{sec:discussion} states the
limitations.

This is version 3 of arXiv:2609.09350. Version 2 \cite{PaperVv2} is the full certified
design study of the same loop; the present version rewrites it around the measurement
protocol, corrects several of its statements (Sec.~\ref{sec:discussion}), and keeps its
certificates in the ancillary files. The use of generative AI in preparing this version is
documented in Appendix~\ref{app:ai}.

\section{Setting}\label{sec:setting}

\subsection{Qubit channels, the Choi test, and monotone readings}\label{sec:choi}

Write a qubit state as $\rho = \tfrac12(\mathbb{I} + \mathbf r\cdot\boldsymbol\sigma)$.
A qubit channel acts as $\mathbf r \mapsto A\mathbf r + \mathbf c$, with a real $3\times3$
matrix $A$ and a real vector $\mathbf c$; we call $(A, \mathbf c)$ the affine pair.
Composition is matrix algebra: $\Phi_2\circ\Phi_1$ has the pair
$(A_2A_1,\, A_2\mathbf c_1 + \mathbf c_2)$. With the maximally entangled state
$\ket{\Omega} = (\ket{00} + \ket{11})/\sqrt2$ on a reference $R$ and the system, the
Choi matrix is $J(\Phi) = (\mathrm{id}\otimes\Phi)(\ket{\Omega}\bra{\Omega})$, and its
partial transpose on $R$ has the closed form
\begin{equation}
  H \;=\; J(\Phi)^{T_R} \;=\; \frac14\Bigl[\,\mathbb{I}\otimes\mathbb{I}
  + \sum_{j} c_j\, \mathbb{I}\otimes\sigma_j
  + \sum_{j,k} A_{jk}\, \sigma_k\otimes\sigma_j\Bigr].
  \label{eq:H}
\end{equation}
$\Phi$ is EB if and only if $J(\Phi)$ is separable \cite{HSR2003}, and for two qubits this
holds if and only if $H \ge 0$ \cite{Peres1996,Horodecki1996}. We therefore read the
single number $\lmin(H)$: negative means the channel preserves some entanglement (NPT),
non-negative means it is EB (PPT). A two-qubit partial transpose has at most one negative
eigenvalue \cite{Sanpera1998}, so $-\lmin$ is the negativity of the Choi state whenever it
is positive. A cheaper condition on the pair alone, $\|A\|_*^2 + |\mathbf c|^2 \le 1$ for
every qubit EB channel, with $\|\cdot\|_*$ the nuclear norm, is necessary but not sufficient.
It is the covariance-matrix criterion \cite{Gittsovich2008} read on the Choi state, and it
is sharp \cite{EBnote}. It can screen candidate channels but cannot replace the test on
$H$.

EB channels form a two-sided ideal: if $\Phi$ is EB then so are $\Psi\circ\Phi$ and
$\Phi\circ\Psi$ for every channel $\Psi$ \cite{HSR2003}. Hence if $\Phi^n$ is EB so is
$\Phi^{m}$ for every $m > n$, and if $\Phi^n$ is not EB neither is $\Phi^m$ for $m < n$.
The readings of a power sequence are monotone: NPT up to round $\nEB - 1$, PPT from round
$\nEB$ on. This is what lets an experiment read the integer from a few rounds, and it is
also a test. Monotonicity survives rounds that differ but are memoryless, because the
$(n+1)$-round channel is then $\Psi\circ\mathcal E_n$ for some channel $\Psi$. On
hardware, however, each $n$ is a separate circuit. A revival of entanglement at a later
round therefore witnesses memory carried between rounds, a non-Markovian effect
\cite{Rivas2010,Rivas2014}, or a difference between circuits of different depth, such as
drift or depth-dependent crosstalk.

\subsection{Collision channels with diagonal ancillas}\label{sec:collision}

One round lets the system qubit $M$ interact with $m$ ancillas through a fixed unitary $U$,
after which the ancillas are discarded (or reset and reused):
\begin{equation}
  \Phi(\rho) \;=\; \Tr_{\mathrm{anc}}\bigl[\,U\bigl(\rho\otimes\tau^{\otimes m}\bigr)U^\dagger\bigr],
  \qquad
  \tau \;=\; \tau_S(p) \;=\; \tfrac12\bigl(\mathbb{I} + p\,S\bigr),
  \label{eq:round}
\end{equation}
where $S$ is a Pauli operator, the bath axis, and $p\in[-1,1]$ the polarisation. For a
bath of splitting $\hbar\omega$ at temperature $T$ along $S = Z$, $p = \tanh(\hbar\omega/2kT)$.
The entries of $(A, \mathbf c)$ are polynomials of degree $m$ in $p$. The key property
of diagonal ancillas is the following.

\begin{lemma}[exact programming of the bath]\label{lem:mix}
Let $\ket{0_S}, \ket{1_S}$ be the eigenbasis of $S$, with $\tau_S(p) =
q_0\ket{0_S}\bra{0_S} + q_1\ket{1_S}\bra{1_S}$, $q_{0,1} = (1\pm p)/2$. For a
configuration $\mathbf b \in \{0,1\}^{mn}$ of the ancilla basis states over $n$ rounds,
let $\mathcal E_{n,\mathbf b}$ be the $n$-round channel on $M$ with those pure ancilla
states. Then
\begin{equation}
  \Phi^n \;=\; \sum_{\mathbf b} w_p(\mathbf b)\, \mathcal E_{n,\mathbf b},
  \qquad
  w_p(\mathbf b) \;=\; \prod_{i=1}^{mn} q_{b_i}.
  \label{eq:mix}
\end{equation}
\end{lemma}

\begin{proof}
With fresh (or perfectly reset) ancillas, the $n$-round channel is
$\rho\mapsto\Tr_{\mathrm{anc}}[V(\rho\otimes\tau^{\otimes mn})V^\dagger]$ with $V$ the
product of the $n$ round unitaries on disjoint ancilla sets. It is linear in the ancilla
state, and $\tau^{\otimes mn} = \sum_{\mathbf b} w_p(\mathbf b)\,
\ket{\mathbf b_S}\bra{\mathbf b_S}$.
\end{proof}

On hardware, each configuration is a basis-state preparation, one single-qubit rotation per
ancilla. The expectation values of every $\mathcal E_{n,\mathbf b}$ are estimated, and the
analysis combines them with the exact weights $w_p(\mathbf b)$. The polarisation therefore
carries no sampling error and needs no calibration beyond that of the basis-state
preparation. At $p = \pm1$ one configuration suffices. At $p = 0$ all $2^{mn}$
configurations enter with equal weight. In between, shots are allocated to configurations
in proportion to their weights, with at least one replica each. Exact weighting needs all
$2^{mn}$ configurations and is practical only for small $mn$ (at most $2^6 = 64$ in the
worked example). For many rounds one draws configurations at random with probabilities
$w_p(\mathbf b)$, which prepares $\tau^{\otimes mn}$ exactly on average but adds sampling
error to the weights. The same reasoning covers
any schedule of polarisations $(p_1,\dots,p_n)$ across rounds, since
$\Phi_{p_n}\circ\cdots\circ\Phi_{p_1}$ has the pair $A(p_n)\cdots A(p_1)$ and
$\sum_j A(p_n)\cdots A(p_{j+1})\mathbf c(p_j)$.

\subsection{Two noise mechanisms}\label{sec:noise}

Two imperfections change the channel in a way the analysis can follow exactly.

\emph{Depolarising noise on $M$.} If each round is followed by
$\rho\mapsto(1-\lambda)\rho + \lambda\,\mathbb{I}/2$, the pair becomes
$((1-\lambda)A, (1-\lambda)\mathbf c)$, again a channel of the same kind. A single final
depolarisation maps $H$ to $(1-\lambda)H + \lambda\,\mathbb{I}/4$ and so raises $\lmin$,
since $\lmin(H)\le 1/4$. Under per-round depolarisation the powers of the noisy round are
still monotone in $n$, but it is not automatic that $\lmin$ at fixed $n$ increases with
$\lambda$. We check this at every target.

\emph{Preparation and reset error of the ancillas.} Suppose the reset leaves an ancilla in
the wrong basis state with probability $e$, independently of the intended state. The
prepared state is then $\tau_S((1-2e)p)$ instead of $\tau_S(p)$, so the channel is the
ideal channel at the effective polarisation
\begin{equation}
  p_{\mathrm{eff}} \;=\; (1-2e)\,p .
  \label{eq:peff}
\end{equation}
Two consequences are used below. The unpolarised bath $p = 0$ is immune to reset error,
which makes it a control. At $p\ne0$ the first-order change is
$\Delta\lmin \approx -2e\,p\,\partial\lmin/\partial p$, whose sign varies from channel to
channel. The model assumes that a reset fails independently of the state before the reset.
That state is correlated with $M$, so a state-dependent failure would couple the rounds.
The composed-versus-direct test of Sec.~\ref{sec:protocol} could detect it.

\section{The protocol}\label{sec:protocol}

The protocol fixes, before any data are taken, a target channel (the circuit $U$, the bath
axis $S$ and polarisation $p$), the set of rounds $\mathcal N$ to be measured, the shots
per setting, and the decision threshold. It then runs six steps.

\paragraph{1. Circuits.} For every $n \in \mathcal N$ and every configuration $\mathbf b$
with its shot allocation, prepare $M$ in one of the six Pauli eigenstates
$P_s\in\{z_\pm, x_\pm, y_\pm\}$, apply $n$ compiled rounds, with the ancillas reset and
re-prepared between rounds, and measure $M$ in the $x$, $y$ or $z$ basis
(Fig.~\ref{fig:circuit}). This gives $18$ settings per configuration. Each job also
contains a readout calibration of $M$ (prepare $\ket0$ and $\ket1$, measure). The pilot
job also contains a reset check of each ancilla qubit: prepare $\ket1$, reset, measure.
The fraction of outcomes $1$ includes the ancilla readout error. The pilot therefore also
calibrates the readout of each ancilla and corrects the fraction for it.

\begin{figure}[t]
\centering
\includegraphics[width=\textwidth]{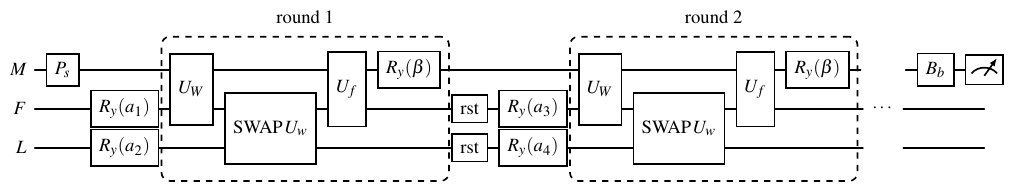}
\caption{Circuits of the protocol, shown for the worked example of
Sec.~\ref{sec:example}. The message qubit $M$ is prepared in a Pauli eigenstate $P_s$,
passes $n$ rounds, and is measured after a basis change $B_b$. The ancillas $F$ and $L$
receive a basis-state preparation $R_y(a_i)$ in every round and a mid-circuit reset
($\mathrm{rst}$) before every round after the first. The angles are $a_i\in\{0,\pi\}$ for a
$Z$ bath and $a_i\in\{\pm\pi/2\}$ for an $X$ bath, fixed by the configuration $\mathbf b$
and bound at run time. On a heavy-hex path the three blocks cost $2$, $3$ and $3$ CZ gates
(Sec.~\ref{sec:compile}).}
\label{fig:circuit}
\end{figure}

\paragraph{2. Estimation.} Invert the readout of $M$ with the calibration. For each
setting combine the configurations with the exact weights of Lemma~\ref{lem:mix}. This
gives the output Bloch vectors $\mathbf r_s$ for the six inputs, and linear inversion
\cite{ChuangNielsen1997} gives
\begin{equation}
  \mathbf c_n = \tfrac12(\mathbf r_{z+} + \mathbf r_{z-}),\qquad
  A_n\mathbf e_a = \tfrac12(\mathbf r_{a+} - \mathbf r_{a-}),\quad a\in\{x,y,z\},
  \label{eq:inversion}
\end{equation}
and then $\lmin$ of Eq.~\eqref{eq:H}. A parametric bootstrap \cite{Efron1994} gives its
standard deviation $\sigma$. It resamples every count, including the readout
calibration, as binomial with the observed frequency, and uses $1000$ replicas. Linear
inversion can return a pair slightly outside the set of channels. That does no harm here,
because $\lmin$ is a statistic, not a state.

\paragraph{3. Readings.} Each block $(S, p, n)$ is read as NPT if $\lmin \le -5\sigma$,
PPT if $\lmin \ge 5\sigma$, and unresolved otherwise. The threshold is fixed in advance.

\paragraph{4. Index.} By the ideal property (Sec.~\ref{sec:choi}) the index is $k$ when
round $k-1$ reads NPT and round $k$ reads PPT. Rounds below an NPT reading and above a
PPT reading need not be measured. If a round between the last NPT and the first PPT is
unresolved, the result is the bracket ``$k$ to $k'$''. A PPT reading below an NPT reading
is inconsistent with a power sequence. It stops the analysis and is reported as evidence
of memory between rounds or of circuit-to-circuit differences (Sec.~\ref{sec:choi}).

\paragraph{5. Composed versus direct.} Round 1 starts from freshly initialised ancillas
and later rounds start from reset ones. The prediction for $n$ rounds is therefore the
measured single-round map $\hat\Phi_1$, with pair $(A_1, \mathbf c_1)$, followed by $n-1$
copies of the shifted map $\tilde\Phi_1$, with pair $(A_1 + \Delta A, \mathbf c_1 + \Delta\mathbf c)$.
Here $(\Delta A, \Delta\mathbf c)$ is the model change of the single-round pair from $p$ to
$p_{\mathrm{eff}}$ at the measured reset error. For each $n\ge2$ compute
\begin{equation}
  D_n \;=\; \lmin\bigl[\text{direct $n$-round estimate}\bigr]
  \;-\; \lmin\bigl[\tilde\Phi_1^{\,n-1}\circ\hat\Phi_1\bigr],
  \label{eq:Dn}
\end{equation}
with the standard deviation from the two bootstraps. The pre-registered pass criterion is
$|D_n| < 3\sigma_D$ for every block. A failure means that the $n$-round channel is not the
composition of single rounds as modelled. The possible causes are memory carried across
the reset, a reset failure that depends on the state, heating by repeated resets,
crosstalk that grows with depth, or drift. A repeat of selected blocks later in
the session separates drift from the rest. In either case the index readings stand,
because they come from the direct measurement.

\paragraph{6. Pilot and shots.} With $18$ settings, a readout error near one per cent,
and $S$ shots per setting, the tomographic standard deviation of $\lmin$ is close to
\begin{equation}
  \sigma \;\approx\; \kappa/\sqrt{S},\qquad \kappa \approx 0.24\ \text{to}\ 0.33,
  \label{eq:sigma}
\end{equation}
across all blocks of the worked example, noisy or ideal. Reading a margin $m$ at $5\sigma$
with $95$ per cent power needs an expected $|z| \ge 5 + 1.645$, hence
\begin{equation}
  S_{\mathrm{req}} \;\approx\; \bigl(6.645\,\kappa/m\bigr)^2 \;\approx\; 3.5/m^2 \quad (\kappa = 0.28).
  \label{eq:Sreq}
\end{equation}
A margin of $0.03$ costs about $4\times10^{3}$ shots per setting, and a margin of $0.005$
about $1.4\times10^{5}$. Because the device's margins are not known in advance, a short
pilot measures the single-round maps first, together with the readout and reset errors.
It predicts each registered reading at $n\ge2$ from the measured single-round map. Round 1
uses the map as measured. Later rounds use it shifted by the model change from $p$ to
$p_{\mathrm{eff}}$ of Eq.~\eqref{eq:peff}, since they start from reset ancillas. The
pilot then applies a fixed rule:
\begin{itemize}
\item a block whose reading is not secured at the registered shots is raised to the
      smallest power of two that secures it, up to a cap. A non-core reading that the
      prediction cannot secure even at the cap is raised to the cap anyway, because the
      prediction inherits the tomographic error of the single-round map, which is of order
      $\sigma_D$;
\item the run goes ahead (GO) if every core reading is secured and the raised run fits
      the budget;
\item otherwise (NO-GO) the pilot is repeated on another qubit layout. Targets and
      thresholds are never changed after data have been seen.
\end{itemize}

\section{Design rules}\label{sec:design}

\subsection{A step has two sides}\label{sec:twosided}

To read $\nEB = k$ an experiment must resolve round $k-1$ as NPT and round $k$ as PPT. The
decision margin of a step is therefore the pair
\begin{equation}
  m_- = -\lmin(H_{k-1}), \qquad m_+ = \lmin(H_k),
  \label{eq:margins}
\end{equation}
and the shot cost is set by the smaller of the two through Eq.~\eqref{eq:Sreq}. Version 2
of this record \cite{PaperVv2} listed $m_-$, with only a sign certificate for the PPT side. Table~\ref{tab:twosided} adds the
certified $m_+$ for its twelve decision points and prices both sides. Only two points have
both sides readable below $10^4$ shots per setting. Seven need $10^5$ to $7\times10^6$ on
the PPT side, even on ideal hardware. At every listed point the PPT side is the thinner one. The primary
target of version 2, $g = 0.85$ at $p = 0.43$, is the most lopsided: $m_- = 0.109$ but
$m_+ = 1.06\times10^{-3}$.

\begin{table}[t]
\caption{Two-sided decision margins of the staircase targets of version 2 \cite{PaperVv2}
(coupling family of Sec.~\ref{sec:loop}; grid polarisations as stored in binary64) and of
the worked example (last three rows). $m_\mp$ are certified (two-sided enclosures of
relative width about $10^{-9}$; values rounded as shown). $S_\mp$ are the shots per setting
needed to read each side at $5\sigma$ with $95$ per cent power on the ideal channel
(float64 tomography simulation, readout error $0.01$).}
\label{tab:twosided}
\centering
\footnotesize
\begin{tabular}{@{}llccccrr@{}}
\toprule
$g$ & bath & $p$ & $k$ & $m_-$ & $m_+$ & $S_-$ & $S_+$\\
\midrule
0.85 & $Z$ & 0.430 & 3 & 0.1086 & 0.00106 & 280 & 3\,490\,760\\
0.85 & $Z$ & 1     & 4 & 0.0414 & 0.01480 & 1\,830 & 13\,002\\
0.70 & $Z$ & 0.747 & 5 & 0.0519 & 0.00452 & 1\,320 & 185\,635\\
0.70 & $Z$ & 1     & 6 & 0.0328 & 0.00177 & 2\,917 & 922\,306\\
0.55 & $Z$ & 0.683 & 8 & 0.0342 & 0.00073 & 3\,265 & 7\,022\,614\\
0.55 & $Z$ & 0.873 & 9 & 0.0246 & 0.00188 & 5\,832 & 968\,221\\
1    & $Z$ & 1     & 3 & 0.0452 & 0.03709 & 1\,653 & 2\,265\\
1    & $X$ & 0.873 & 3 & 0.0953 & 0.00241 & 366 & 542\,934\\
1    & $X$ & 1     & 4 & 0.0563 & 0.00443 & 952 & 142\,784\\
1    & $Y$ & 0.620 & 3 & 0.0883 & 0.00672 & 448 & 70\,746\\
1    & $Y$ & 0.747 & 4 & 0.0315 & 0.02201 & 3\,097 & 5\,777\\
1    & $Y$ & 0.873 & 5 & 0.0249 & 0.00565 & 4\,213 & 67\,971\\
\midrule
1.2  & $X$ & 1     & 3 & 0.0721 & 0.00561 & 569 & 89\,746\\
1.2  & $Z$ & 1     & 2 & 0.1146 & 0.03623 & 264 & 2\,345\\
1.2  & any & 0     & 2 & 0.0925 & 0.09303 & 436 & 495\\
\botrule
\end{tabular}
\end{table}

\subsection{Noise thickens one side and thins the other}\label{sec:asymmetry}

At every target of the worked example, per-round depolarising noise and reset error raise
$\lmin$ (checked in float64 on a grid of $\lambda\le0.12$ and $e\le0.1$).
The NPT side of a step is eaten, and the PPT side thickens. The two sides of
Table~\ref{tab:twosided} therefore play different roles on hardware. A target should have
a large ideal $m_-$, because that side is fragile. It can live with a thin ideal $m_+$, as
long as the expected noise or a raised shot count covers it. The contrast between two
index-3 targets makes the point. At $g = 1$ on the $Z$ bath both ideal sides are
comfortable ($0.045$ and $0.037$), but in the model of Sec.~\ref{sec:tolerance} the NPT
side survives only a per-round contraction up to $0.025$ (at $4096$ shots and $e = 0$). At $g = 1.2$ on the $X$ bath the
ideal PPT side is thin ($0.0056$), but the NPT side survives up to $0.07$, and noise
thickens the thin side.

\subsection{Tolerance of the worked example}\label{sec:tolerance}

Table~\ref{tab:tolerance} gives, for the readings of the worked example, the range of
per-round depolarising contraction $\lambda$ over which each reading holds with $95$ per
cent power at $5\sigma$. The model includes the reset error $e$ of
Eq.~\eqref{eq:peff}, readout error $0.01$ and a one-time preparation-and-measurement
contraction of $0.01$. Two things stand out. The only fragile core reading is the X-bath
NPT reading at $n=2$, and reset error is what makes it fragile: the ideal X-bath index
drops from 3 to 2 once $p_{\mathrm{eff}}$ falls below about $0.85$. And the $Z$-bath and
$p = 0$ readings hold over the whole scanned range. The model is isotropic and therefore
pessimistic. In the stressed simulation of Sec.~\ref{sec:sims}, with a fitted contraction
of $0.052$ and reset errors near $0.01$, the X-bath $n=2$ reading still came out NPT at
$-10.7\sigma$. The pilot rule for this reason works from measured single-round maps, not
from this table.

\begin{table}[t]
\caption{Noise tolerance of the readings of the worked example (planning model, float64).
Entries give the range of per-round depolarising contraction $\lambda$ (scanned from $0$
to $0.12$ in steps of $0.005$) over which the reading holds with $95$ per cent power at
$5\sigma$, for ancilla reset error $e$. ``All'' means the whole scanned range. The model
applies $e$ to every round, including the first. That is conservative for the NPT rows
and optimistic for the PPT rows.}
\label{tab:tolerance}
\centering
\footnotesize
\begin{tabular}{@{}lcccc@{}}
\toprule
reading (shots per setting) & $e=0$ & $e=0.01$ & $e=0.02$ & $e=0.03$\\
\midrule
$X$, $p=1$, $n=2$ NPT (4096) & $\le 0.070$ & $\le 0.050$ & $\le 0.035$ & $\le 0.020$\\
$X$, $p=1$, $n=3$ PPT (4096) & $\ge 0.035$ & $\ge 0.020$ & all & all\\
$X$, $p=1$, $n=3$ PPT (16384) & $\ge 0.010$ & all & all & all\\
$Z$, $p=1$, $n=1$ NPT and $n=2$ PPT (4096) & all & all & all & all\\
$p=0$, $n=1$ NPT and $n=2$ PPT (4096) & all & all & all & all\\
\midrule
$g=0.85$, $Z$, $p=1$, $n=3$ NPT (16384) & $\le 0.030$ & $\le 0.025$ & $\le 0.020$ & $\le 0.020$\\
$g=0.85$, $Z$, $p=1$, $n=4$ PPT (4096) & $\ge 0.015$ & $\ge 0.010$ & $\ge 0.005$ & all\\
$g=0.85$, $Z$, $p=0.43$, $n=3$ PPT (8192) & $\ge 0.030$ & $\ge 0.030$ & $\ge 0.030$ & $\ge 0.030$\\
\botrule
\end{tabular}
\end{table}

The last three rows show why the $g = 0.85$ staircase ($3$ at $p = 0.43$, $4$ at $p=1$) is
not a good first target. With perfect resets its $p = 1$ readings hold for $\lambda$
between about $0.015$ and $0.030$, and its $p = 0.43$ PPT reading needs $\lambda \ge 0.030$:
the whole staircase is readable only in a narrow window near $\lambda \approx 0.03$. With
$e \ge 0.01$ it is readable at no noise level.

\section{Worked example}\label{sec:example}

\subsection{The round}\label{sec:loop}

The channel is one round of a three-qubit loop in which the message qubit $M$ meets two
ancillas, $F$ and $L$ ($m = 2$), both prepared in $\tau_S(p)$ \cite{PaperI}. With
$R_y(a) = e^{-iaY/2}$ and registers ordered $(M, F, L)$,
\begin{equation}
  U \;=\; \bigl(R_y(\beta)\otimes\mathbb{I}\otimes\mathbb{I}\bigr)\,U_f(\varphi)\,U_w(\kappa)\,U_W(\theta),
  \label{eq:U}
\end{equation}
where
\begin{align*}
  U_W(\theta) &= \ket0\bra0_M\otimes R_y(\pi-2\theta)_F + \ket1\bra1_M\otimes R_y(2\theta)_F,\\
  U_w(\kappa) &= \exp(-i\tfrac{\kappa}{2} Z_F\otimes Y_L),\\
  U_f(\varphi) &= \exp(-i\varphi\,\mathrm{SWAP}_{ML}).
\end{align*}
$U_W$ writes $M$ into $F$ by a controlled rotation, $U_w$ couples $F$ to $L$, and the
partial swap $U_f$ returns part of $L$ to $M$. We use the one-parameter family
\begin{equation}
  \theta = \tfrac{\pi}{4} + g\,(\theta_0 - \tfrac{\pi}{4}),\quad \varphi = g\varphi_0,\quad
  \kappa = g\kappa_0,\quad \beta = \beta_0,
  \label{eq:family}
\end{equation}
with $(\theta_0,\varphi_0,\kappa_0,\beta_0) = (0.16345853,\, 0.20061939,\,
0.4323098,\, 0.23903823)\times\pi$, the certified reference point of
Refs.~\cite{PaperI,PaperII}, and take $g = 1.2$. At the reference point ($g = 1$) with a $Z$
bath the index is $3$ at every polarisation $p\in[0,1]$, certified over the continuum in
Ref.~\cite{PaperII}. That reference also studies the infinite-temperature floor of this
loop and its thermal valleys.

\subsection{Targets}\label{sec:targets}

Table~\ref{tab:targets} lists the certified ideal values of $\lmin$ for three baths. The
certification uses $256$-bit ball arithmetic \cite{Arb2017}. The upper bound is the
Rayleigh quotient of a proposed eigenvector. The lower bound is Sylvester's criterion on
$H - s\,\mathbb{I}$, which proves $\lmin > s$ (Appendix~\ref{app:cert}).
\begin{itemize}
\item The $X$ bath at $p = 1$ has index $3$ and is the headline reading.
\item The $Z$ bath at the same population has index $2$. The contrast between the two is
      a measurable instance of the dependence of the index on bath direction discussed in
      version 2.
\item The unpolarised bath has index $2$ and is the reset-immune control.
\end{itemize}
The core readings are X-bath $n = 2$ NPT and, for the other two baths, $n=1$ NPT and
$n=2$ PPT. The X-bath $n=3$ PPT reading completes the headline, and X-bath $n=4$ is a
backstop that brackets the index if $n = 3$ is unresolved.

\begin{table}[t]
\caption{Certified ideal $\lmin$ of the $n$-round channel at $g = 1.2$ (enclosures of
relative width about $10^{-9}$, rigorous but not tight; six digits shown).}
\label{tab:targets}
\centering
\footnotesize
\begin{tabular}{@{}lccccc@{}}
\toprule
bath & $n = 1$ & $n = 2$ & $n = 3$ & $n = 4$ & $\nEB$\\
\midrule
$X$, $p = 1$ & $-0.221643$ & $-0.072137$ & $+0.005611$ & $+0.043112$ & 3\\
$Z$, $p = 1$ & $-0.114649$ & $+0.036231$ & $+0.091629$ & & 2\\
$p = 0$      & $-0.092454$ & $+0.093030$ & $+0.178514$ & & 2\\
\botrule
\end{tabular}
\end{table}

\subsection{Compilation}\label{sec:compile}

The three two-qubit blocks have analytic local-equivalence classes
\cite{Makhlin2002,Zhang2003}. $U_W$ and $U_w$ each have a single nonzero canonical
coordinate, $|\pi/4-\theta|$ and $\kappa/2$ (reduced to $\pi/2-\kappa/2$ when
$\kappa/2 > \pi/4$, as at $g = 1.2$), and need exactly two \cnot{}s. The partial swap
has three equal nonzero coordinates $\varphi/2$ and needs three \cnot{}s
\cite{VatanWilliams2004}. A round therefore costs seven \cnot{}s with all-to-all
connectivity, or five native arbitrary-angle two-qubit gates where these exist.

On a heavy-hex lattice \cite{Chamberland2020} the three qubits sit on a path $M$--$F$--$L$
with $F$ in the middle, and $U_f$ acts on the two ends. Routing it would cost a SWAP. The
SWAP can instead be merged into the preceding block: $\mathrm{SWAP}_{FL}\,U_w$ is a
three-\cnot{} gate. After it the state of $L$ sits on the middle qubit, and $U_f$ acts on
two neighbours. The compiled round is
$\mathrm{SWAP}_{FL}\cdot U$: $M$ never moves, and the swap of the two ancillas is
irrelevant because both are reset before the next round. The cost is $2 + 3 + 3 = 8$ CZ
gates per round. Version 2 stated ten, which counts the SWAP separately.

We verified the compiled circuits exactly in noiseless simulation (Qiskit
\cite{Qiskit2024}), resets included. For $n$ up to 3 the channel on $M$ agrees with the
model to about $10^{-16}$. Transpiled for the heavy-hex device ibm\_kingston (156-qubit
Heron r2), each round costs eight CZ gates. On the path $(21, 22, 23)$, before echo
padding, the three-round circuit has depth $113$ with four mid-circuit resets. An X--X echo \cite{ViolaKnillLloyd1999} is padded on $M$ in
idle windows, mainly during the resets. The padded, parametric circuits still reproduce
the exact channel to about $10^{-15}$. Mid-circuit reset \cite{Corcoles2021} keeps the
qubit count at three for any $n$. Fresh ancillas would cost $2n+1$ qubits.

\subsection{Budget}\label{sec:budget}

The run has three jobs:
\begin{itemize}
\item a pilot J1: $n = 1$ for the three baths, plus readout and reset checks;
\item the main job J2: $n = 2,3,4$ for the $X$ bath and $n = 2,3$ for the other two;
\item a repeat J3 of the $X$-bath $n = 2, 3$ blocks later in the same session.
\end{itemize}
All blocks of J1 to J3 start at $4096$ shots per setting (Table~\ref{tab:budget}).

The processor time is estimated with a linear cost model, fitted to the author's earlier
jobs on the same device: a per-binding overhead of $5.7$~ms plus $0.27$~ms per shot, with
a $5$ per cent margin. The three jobs cost $273$~s, or $397$~s with the raise the pilot
asked for in the nominal simulation. That is a few hundred seconds of processor time. An optional
job J4 for the $g = 0.85$ staircase costs a further $199$~s.

\begin{table}[t]
\caption{Jobs of the worked example at the registered shots ($4096$ per setting, except
in J4, which uses $16384$ for $p=1$, $n=3$ and $8192$ for $p=0.43$, $n=3$). A
binding is one circuit with one set of ancilla angles. The QPU seconds are estimates from
the cost model of Sec.~\ref{sec:budget}.}
\label{tab:budget}
\centering
\footnotesize
\begin{tabular}{@{}llrrr@{}}
\toprule
job & content & bindings & shots & QPU s\\
\midrule
J1 pilot & $n=1$, three baths; readout of $M$, $F$, $L$; reset check & 116 & 253\,952 & 72\\
J2 main & $X$: $n=2,3,4$; $Z$ and $p=0$: $n=2,3$; readout & 1\,532 & 524\,288 & 157\\
J3 repeat & $X$: $n=2,3$; readout & 38 & 155\,648 & 44\\
\midrule
J4 optional & $g=0.85$, $Z$: $p=1$, $n=3,4$; $p=0.43$, $n=2,3$ & 3\,890 & 623\,360 & 199\\
\botrule
\end{tabular}
\end{table}

\subsection{Simulated runs}\label{sec:sims}

We ran the full pipeline on the Qiskit noise model of ibm\_kingston
(FakeKingston, qiskit-ibm-runtime 0.49), on the path $(140, 141, 142)$. That path ranked
first under a per-round cost proxy computed from a calibration snapshot of the device. Two
scenarios, meant to fall on either side of the device's actual noise, test the rules:
\begin{itemize}
\item \emph{nominal}: the noise model as shipped;
\item \emph{stressed}: in addition, two-qubit depolarising noise of $0.004$ on every CZ of
      the path and a bit flip of probability $0.01$ after every reset.
\end{itemize}
Each scenario passed through the same pilot rule and analysis as a hardware run would
(Table~\ref{tab:sims}, Fig.~\ref{fig:staircase}).

\begin{table}[t]
\caption{Simulated runs (FakeKingston noise model, path $(140,141,142)$). The pilot
columns give the fitted single-round contraction $\eta$, the measured reset errors of $F$
and $L$, and the pilot decision. The index columns give the readings at $5\sigma$. The
reset errors are corrected for the ancilla readout. The last column is the largest
$|D_n|/\sigma_D$ over all $n\ge2$ blocks. The stressed scenario ran J1 and J2 only.}
\label{tab:sims}
\centering
\footnotesize
\setlength{\tabcolsep}{4pt}
\begin{tabular}{@{}lccccccc@{}}
\toprule
scenario & $\eta$ & $e_F$, $e_L$ & pilot & $X$, $p=1$ & $Z$, $p=1$ & $p=0$ & max $|D_n|$\\
\midrule
nominal, 4096 shots & $-0.001$ & 0.005, 0.000 & GO, raise & 3 to 4 & 2 & 2 & $1.4\sigma$\\
nominal, raised & $-0.001$ & 0.005, 0.000 & GO & 3 & 2 & 2 & $1.9\sigma$\\
stressed & 0.052 & 0.008, 0.012 & GO & 3 & 2 & 2 & $2.1\sigma$\\
\botrule
\end{tabular}
\end{table}

\begin{figure}[t]
\centering
\includegraphics[width=\textwidth]{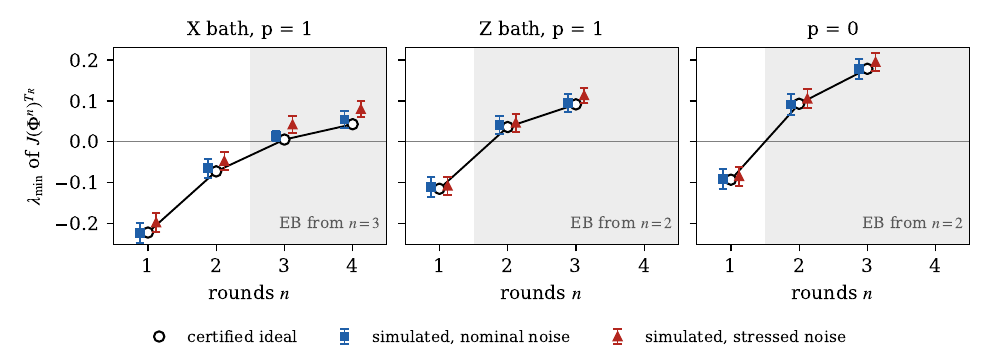}
\caption{Certified ideal $\lmin$ (open circles) and simulated readings with
$\pm5\sigma$ bars, for nominal noise with the pilot's raise (squares) and stressed noise
(triangles). The shaded region is where the ideal channel is EB. A reading counts as PPT
or NPT when its whole bar lies above or below zero.}
\label{fig:staircase}
\end{figure}

\emph{Nominal.} The pilot found a contraction consistent with zero and predicted the
$X$-bath $n = 3$ PPT reading at only $1.6\sigma$ at $4096$ shots, too thin to secure even
at the cap. As the rule prescribes for a non-core reading, it raised that block and its
repeat to the cap of $16384$. That brings J2 to $219$~s and J3 to $106$~s, $397$~s in all.
At the registered shots the $n = 3$ block read $+3.0\sigma$ (unresolved), and the index
was the bracket ``3 to 4''. With the raise it read $+5.4\sigma$, the repeat read
$+7.9\sigma$, and the index was 3. The measured margin, $0.014$, was twice the pilot's
prediction of $0.006$, a difference within the prediction's own uncertainty. The weakest
reading of the headline sits at low noise, as Sec.~\ref{sec:asymmetry} predicts.

\emph{Stressed.} The fitted contraction was $0.052$. The reset errors were $0.008$ and
$0.012$ after the readout correction, against $0.024$ and $0.016$ as read. No raise was
needed. The $X$-bath readings were $-10.7\sigma$ at $n = 2$ and $+10.3\sigma$ at $n = 3$,
giving index 3. Noise had thickened the thin side and eaten
part of the robust one, and both still read.

\emph{Composed versus direct.} In both scenarios no $|D_n|$ exceeded $2.1\sigma$. The
simulated rounds are memoryless and differ only by the modelled reset error, so this is
the expected result. The reset shift matters: without it, the stressed X-bath blocks at
$n = 3$ and $4$ would have given $D_n$ of $+1.5\sigma$ and $+2.1\sigma$, against $+0.4\sigma$
and $+0.9\sigma$ with it, while at $n = 2$ the shift moved $D_n$ from $-0.6\sigma$ to
$-1.4\sigma$. The uncertainty of the measured $e$ is not included in $\sigma_D$. It adds
about $0.002$ to $D_n$ here. The test is a null test here. On hardware it is a measurement.

\emph{The $g = 0.85$ staircase.} Run under nominal noise, J4 read both NPT sides clearly
($-11.9\sigma$ at $p=1$, $n=3$, and $-21.6\sigma$ at $p=0.43$, $n=2$). Both PPT sides
were unresolved ($+3.4\sigma$ and $+3.2\sigma$), so it returned only lower bounds on the
two indices, as Table~\ref{tab:tolerance} predicts.

These simulations test the pipeline and the rules, not the hardware. What they show is
that the protocol, run exactly as registered, recovers the certified indices under two
noise levels chosen to fall on either side of what current devices show, in a model without coherent
errors, and that the pilot identifies the one reading
that needs more shots.

\section{Beyond static powers}\label{sec:beyond}

Because composition is exact for any schedule of polarisations (Sec.~\ref{sec:collision}),
the same machinery certifies the index under time-dependent baths. Version 2
\cite{PaperVv2} explores this for the loop of Sec.~\ref{sec:loop} and deposits the
certificates, which are kept in the ancillary files of this version. We summarise, because
it marks where the direct measurement stops being practical.

\emph{Thermal valleys.} Some circuits break entanglement sooner at an intermediate bath
polarisation than at either extreme; the first certified examples are in
Ref.~\cite{PaperII}. The gated random search of version 2, over $33$ million circuits, found
$10{,}628$ such valleys, and none reaches a contrast of $0.05$. The contrast has an
interior maximum near index $90$, at a certified $2.85\times10^{-3}$.

\emph{Control results.} Exact inhomogeneous composition turns a valley into a switch. At
the best valley, held inside its window at $n = 88$ rounds:
\begin{itemize}
\item a centred pulse of exactly eight rounds restores entanglement;
\item a $53$-round pulse beats the static protocol by a certified margin;
\item a $1.2$ per cent coherent tilt of the ancilla opens the switch at fixed population.
\end{itemize}

\emph{Cost.} These signals lie between about $10^{-4}$ (the minimal eight-round pulse)
and $10^{-2}$ (the tilted ancilla), mostly near $3\times10^{-3}$, over $80$ to $110$
rounds. By Eq.~\eqref{eq:Sreq} a margin of $3\times10^{-3}$ costs about $4\times10^{5}$
shots per setting for each of the $18$ settings. Version 2 estimated about $3\times10^{6}$,
from an assumed error model for two-qubit state tomography. At intermediate polarisation,
exact weighting over $2^{176}$ configurations is out of the question, and configurations
must be sampled. Within a per-round depolarising model the signals die at a per-round
contraction near $10^{-4}$. They are certified statements about the
model and a map of what better hardware could reach, not near-term measurement targets.

\section{Discussion and limitations}\label{sec:discussion}

\emph{What is certified, and what is not.} The target values of Tables~\ref{tab:twosided}
and \ref{tab:targets} are certified statements about the ideal channel. The shot
requirements, the tolerance table and the simulated runs are float64 planning estimates.
A hardware run measures the index of the channel it implements. Depolarising-type noise
and reset error move that index down or leave it unchanged at the targets of the worked
example, but coherent errors can move $\lmin$ either way. The certified ideal index is
therefore a comparison value, not a pass criterion.

\emph{Scope of the statistics.} The bootstrap covers shot noise and readout calibration,
not drift or systematic preparation error. Drift is addressed by the in-session repeat and
the composed-versus-direct test. Preparation and measurement errors of $M$ that act as
channels compose the channel with a noisy channel before and after it. By the ideal
property they can turn an NPT reading into PPT but never the reverse, so, like the other
noise, they bias the reading towards a smaller index. The readout inversion removes most
of the measurement part. An over- or under-correction in that inversion is not a channel,
however, and could push the reading either way. The reset check bounds the ancilla reset error from above unless the ancilla readout
is calibrated too, as it is in the pilot of the worked example. Maximum-likelihood tomography, or a
tomography method with certified confidence regions, could replace linear inversion
without changing the rest of the protocol.

\emph{Generality.} Nothing in Secs.~\ref{sec:setting} and \ref{sec:protocol} uses the
particular loop. The protocol applies to any qubit collision channel with ancillas
diagonal in a known basis, on any platform with mid-circuit reset, or with $mn + 1$ qubits
without it. For a system of dimension $d > 2$ the Choi matrix is $d^2\times d^2$ and a
positive partial transpose no longer implies separability. The PPT reading then bounds the
index of the PPT property, and an EB test needs a separability criterion or witness.

\emph{Corrections to version 2.} The following statements of version 2 are corrected
here.
\begin{itemize}
\item The heavy-hex count is eight CZ gates per round, not ten (Sec.~\ref{sec:compile}).
      The derived counts change accordingly: $72$ CZ gates, not $90$ \cnot{}s, for the
      deepest staircase target, and $704$, not $880$, for the $88$-round valley.
\item The decision margin of a step has two sides. At every listed target the side that
      version 2 did not list is the thinner one (Table~\ref{tab:twosided}). The shot
      range of version 2, about $2\times10^{3}$ to $4\times10^{4}$ per setting for the
      selected targets, came from its own calibration of the NPT side. In the present
      calibration the NPT sides need $280$ to $5{,}832$ shots per setting, and the PPT
      sides up to $7\times10^{6}$.
\item Exact programming of $p$ holds for fresh ancillas. With mid-circuit reuse a reset
      error shifts $p$ to $(1-2e)p$ (Eq.~\eqref{eq:peff}), so the reset does enter the
      ancilla state.
\item Version 2 said that no experiment had measured an entanglement-breaking index.
      That is not correct as stated. The data of Ref.~\cite{Cuevas2017} fix an index of
      $2$ for a photonic channel. What remains unmeasured is a larger index located by a
      direct scan over rounds.
\end{itemize}

The worked example is ready to run. It needs three qubits, eight CZ gates per round, four
rounds at most, and an estimated few hundred seconds of processor time.

\backmatter

\bmhead{Acknowledgements}
The computations used Arb through python-flint, NumPy, Qiskit and Qiskit Aer.

\section*{Declarations}

\bmhead{Funding} No external funding was received for this work.

\bmhead{Competing interests} The author has no competing interests to declare.

\bmhead{Ethics approval and consent to participate} Not applicable.

\bmhead{Consent for publication} Not applicable.

\bmhead{Data availability} No experimental data were generated. The certificates, the
simulated run records and analyses, and the complete search output of version 2 are in
the ancillary files of arXiv:2609.09350 (Appendix~\ref{app:files}).

\bmhead{Materials availability} Not applicable.

\bmhead{Code availability} The circuit builders, run and analysis scripts, the
certification scripts and the re-derivation script are in the ancillary files of
arXiv:2609.09350. The certification library from which the primitives derive is archived
at Zenodo \cite{Software}.

\bmhead{Author contribution} E.K. is the sole author. He conceived and designed the study, carried out and checked the analysis with the assistance described in Appendix~\ref{app:ai}, and wrote the manuscript.

\bmhead{Use of generative AI} See Appendix~\ref{app:ai}.

\begin{appendices}

\section{Certification}\label{app:cert}

All certified numbers are two-sided enclosures in Arb ball arithmetic \cite{Arb2017} at
$256$ bits, through python-flint. Circuit parameters and polarisations enter as exact
binary64 numbers, so each certificate concerns the circuit as stored. The round unitary is
built from closed-form sines and cosines, since $(Z\otimes Y)^2 = \mathrm{SWAP}^2 =
\mathbb{I}$, and the affine pair by tracing out the ancillas. Composites follow exactly
from the pair algebra of Sec.~\ref{sec:choi}. For $\lmin$ of the $4\times4$ matrix $H$:
\begin{itemize}
\item the upper bound is the Rayleigh quotient of a float64-proposed eigenvector,
      certified in ball arithmetic;
\item the lower bound $s$ comes from Sylvester's criterion on $H - s\,\mathbb{I}$ (all
      leading principal minors certified positive), with $s$ backed off from the float64
      value until the criterion passes;
\item an enclosure that straddles zero counts as a failed sign certificate.
\end{itemize}
No floating-point eigensolver is trusted. The worked-example targets are produced by
\texttt{certify\_g120.py}, the two-sided margins of Table~\ref{tab:twosided} by
\texttt{two\_sided\_margins.py}, and every certified number of version 2 by
\texttt{paperV\_checks.py}.

\section{Deposited files}\label{app:files}

The ancillary directory of arXiv:2609.09350v3 keeps all files of version 2 unchanged:
certificates, search chassis and output, schedule engine, and the re-derivation script
\texttt{paperV\_checks.py}. It adds a folder \texttt{staircase/} for this version:
\begin{itemize}
\item \texttt{staircase\_circuits.py}: circuit builders and exactness checks;
\item \texttt{run\_staircase.py}: build, price, simulate or submit the jobs;
\item \texttt{analyze\_staircase.py}: estimation, bootstrap, readings, index, pilot rule
      and composed-versus-direct test;
\item \texttt{choose\_layout.py}: layout ranking;
\item \texttt{tolerance.py}: Table~\ref{tab:tolerance};
\item \texttt{certify\_g120.py} and \texttt{two\_sided\_margins.py}: the certified targets
      and margins, with their outputs;
\item the records, counts and analyses of every simulated run in
      Sec.~\ref{sec:sims}, and \texttt{make\_fig2.py} for Fig.~\ref{fig:staircase};
\item a \texttt{README.md} with the commands that reproduce every table and figure of
      this version.
\end{itemize}

\section{Methods note: use of generative AI}\label{app:ai}

This version was prepared with Claude (Anthropic), used through the Claude desktop app
between 15 and 23 September 2026 (model versions Claude Fable 5.1 and Claude Opus 5.5). It
was used to:
\begin{itemize}
\item draft and revise the manuscript text;
\item write and run the scripts of the worked example (circuit construction and
      verification, noisy simulations, noise-tolerance scans, analysis, and the new
      certification scripts);
\item check the references against Crossref and arXiv.
\end{itemize}
The author designed the study, directed and checked all of this work, verified the
results, and takes full responsibility for the content.

\end{appendices}

\bibliography{refs}

\end{document}